\documentclass[a4paper]{article}
\usepackage{amsthm,amscd}
\usepackage{amsmath,amssymb,amsfonts}
\usepackage{multirow}
\usepackage{graphicx}               
\usepackage{color}                  
\usepackage[absolute]{textpos} 
\usepackage{multicol}
\usepackage{array}

\usepackage[T1]{fontenc}

\usepackage[utf8]{inputenc}
\usepackage[english]{babel}

\newcommand{\beq} {\begin{eqnarray*}}
\newcommand{\eeq} {\end{eqnarray*}}
\newcommand{\trm} {\textrm}
\newcommand{\tbf} {\textbf}
\newcommand{\noi} {\noindent}

\providecommand{\norm}[1]{\left\Vert#1\right\Vert}
\def \R{\mathbb{R}}

\def \N{\mathbb{N}}

\def \E{\mathbb{E}}
\def \P{\mathbb{P}}
\def \Var{\hbox{{\rm Var}}}

\def \Cov{\hbox{{\rm Cov}}}

\providecommand{\abs}[1]{\left\lvert#1\right\rvert}
\newcommand{\pscal}[1]{\langle #1 \rangle}

\DeclareMathOperator*{\argmin}{Argmin}

\newcommand{\Tr}{\mathrm{Tr}}
\newcommand{\Xu}{X_{\bf u}}
\newcommand{\Cu}{C_{\bf u}}
\newcommand{\Xtu}{X_{\sim\bf u}}
\newcommand{\Ctu}{C_{\sim\bf u}}
\newcommand{\Cutu}{C_{\bf u,\sim\bf u}}
\newcommand{\suf}{S^{\bf u}(f)}

\theoremstyle{plain}
\newtheorem{thm}{Theorem}[section]
\newtheorem{proposition}{Proposition}[section]
\newtheorem{lem}{Lemma}[section]
\newtheorem{cor}{Corollary}[section]

\newtheorem{ex}{Example}[section]
\newtheorem{rmk}{Remark}[section]
\newtheorem{defi}{Definition}[section]

\newcommand{\cro}[1]{\left[{#1}\right]}

\renewcommand{\thefootnote}{\arabic{footnote}}

\begin{document}
\title{Sensitivity analysis for multidimensional and functional outputs}
\footnotetext[1]{Laboratoire de Mathématiques d'Orsay, B\^atiment 425, Université Paris-Sud, 91405 Orsay, France}
\footnotetext[2]{Laboratoire de Statistique et Probabilités, Institut de Mathématiques
		Université Paul Sabatier (Toulouse 3) 31062 Toulouse Cedex 9, France}
\author{\renewcommand{\thefootnote}{\arabic{footnote}}Fabrice Gamboa\footnotemark[2],
Alexandre Janon\footnotemark[1],
Thierry Klein\footnotemark[2],
Agn\`es Lagnoux\footnotemark[2]}

\maketitle

\begin{abstract}
Let  $X:=(X_1, \ldots, X_p)$ be   random objects (the inputs),  defined on some probability space $(\Omega,{\mathcal{F}}, \mathbb P)$ and valued in some measurable space $E=E_1\times\ldots \times E_p$. Further,  let  $Y:=Y = f(X_1, \ldots, X_p)$ be  the output.
Here,  $f$ is a measurable function from $E$ to some Hilbert space $\mathbb{H}$ ($\mathbb{H}$ could be either of finite or infinite dimension).
In this work, we give a natural generalization of the Sobol  indices (that are classically defined when $Y\in\R$ ), when the output belongs to  $\mathbb{H}$. These indices have very nice properties. First, they are invariant. under isometry and scaling. Further they can be, as in dimension $1$, easily estimated by using the so-called Pick and Freeze method.  We investigate the asymptotic behaviour of such estimation scheme.

\end{abstract}

\begin{bf}Keywords: \end{bf}
  Semi-parametric efficient estimation, sensitivity analysis, quadratic functionals, Sobol indices, vector output, temporal output, concentration inequalities.\\
\begin{bf} Mathematics Subject Classification. \end{bf}  
62G05, 62G20

\tableofcontents

\section{Introduction}
\indent
Many mathematical models encountered in applied sciences involve a large number of poorly-known parameters as inputs. It is important for the practitioner to assess the impact of this uncertainty on the modeliid output. An aspect of this assessment is sensitivity analysis, which aims to identify the most sensitive parameters. In other words, parameters that have the largest influence on the output. In global stochastic sensitivity analysis, the input variables are assumed to be  independent random variables.  Their probability distributions  account for the practitioner's belief about the input uncertainty. This turns the model output into a random variable. Using the so-called Hoeffding decomposition \cite{van2000asymptotic}, the total variance of a scalar output can be split down into different partial variances. Each of these partial variances measures the uncertainty on the output induced by the corresponding input variable. By considering the ratio of each partial variance to the total variance, we obtain a measure of importance for each input variable called the \emph{Sobol index} or \emph{sensitivity index} of the variable \cite{sobol1993}; the most sensitive parameters can then be i.dentified and ranked as the parameters with the largest Sobol indices.  A clever way to estimate the Sobol indices is to use the so-called Pick and Freeze sampling scheme (see  \cite{sobol1993} and more recently \cite{janon2012asymptotic}). This sampling scheme transforms the complex initial statistical problem of estimation into a simple linear regression problem. Widely used by practitioners in many applied fields (see for example \cite{saltelli2008global} and the complete bibliography given therein), 
the mathematical analysis of the Pick and Freeze scheme has been recently performed in \cite{janon2012asymptotic} and \cite{pickfreeze} (see also \cite{xu2011} where some mathematical draft ideas are given). In the last decade, many authors have proposed some generalizations of Soboliid indices for scalar outputs (see for example \cite{bobor2007}, \cite{Owen2012},  \cite{Owen2013} and \cite{Fort2013}).   The aim of the present paper is twofold. First, we wish to build some extensions of Sobol indices in the case of vectorial or functional output. Secondly, we aim to construct some Pick and Freeze estimators of such extensions and to study their asymptotic and non-asymptotic properties. Generalization of the Sobol index for multivariate or functional outputs has been considered  in an empirical way in 
\cite{campbell2006sensitivity} and \cite{lamboni2011multivariate}. In this paper, we consider and study a new generalization of Sobol indices for vector or functional  outputs. This generalization was implicitly considered in the pioneering work of Lamboni et al ( \cite{lamboni2011multivariate}). The starting point of the construction of these new indices relies on the multidimensional Hoeffding decomposition of the vectorial output.  Further, due to non-commutativity, many choices for an extension of Sobol indices are possible.  To restrict the choice we both require that the indices satisfy natural invariance properties and remain easy to estimate when using a Pi.ck and Freeze sampling scheme.  

\indent
The paper is organized as follows. To begin with, we start in the next section by developing and discussing two examples. 
These examples illustrate the difficulty of extending directly scalar Sobol indices to a multidimensional context.
The generalized Sobol indices and their main properties are given in Section \ref{sec:definition}. 
In a nutshell, these newiid Sobol indices are the same as the classical ones of the 
unidimensional context, up to the trace operation taken on both terms of the ratio. We show that these quantities are well-tailored for sensitivity 
analysis, as they are invariant under isometry and scaling of the output. 
In Section \ref{sec:uniqueness}, 
we introduce another general family of Sobol matricial indices. They are also compatible with the Hoeffding decomposition and they also satisfy the 
natural  invariance properties. Each element of this family depends on a probability measure on the group of signed permutation matrices. 
The main drawback of these quantities is that they are not so easy to estimate unlike the indices introduced in Section \ref{sec:definition}. In Section 
\ref{sec:estimation} we revisit the Pick and Freeze sampling scheme and study the asymptotic and non-asymptotic properties of the Pick and Freeze estimators of the new indices. These properties are numerically illustrated on two relevant examples in Section \ref{sec:numill}. To finish, the extension of our results, we present in  Section \ref{sec:functional} 
the case of functional outputs.

\section{Motivation}\label{sec:ex_intro}
We begin by considering two examples that enlighten the need for a proper definition of sensitivity indices for multivariate outputs.

\begin{ex}\label{ex0}
Let us consider the following nonlinear model
 $$Y=f^{a,b}(X_{1},X_{2}):=\begin{pmatrix}
f^{a,b}_1(X_1,X_2)\\
f^{a,b}_2(X_1,X_2)
\end{pmatrix}=\begin{pmatrix}
X_{1}+X_{1}X_{2}+X_{2}\\
aX_{1}+bX_{1}X_{2}+X_{2}
\end{pmatrix}$$ 
where $X_1$ and $X_2$ are  assumed to be i.i.d. standard Gaussian random variables  (r.v.s).

\noi\\
First, we compute the one-dimensional Sobol indices $S^{\bf j }(f^{a,b}_i)$ of $f^{a,b}_i$ with respect to $X_j$ ( $i,j=1,2$). We get
\beq
(S^{\bf 1}(f^{a,b}_1),S^{\bf 1}(f^{a,b}_2))&=&(1/3,a^2/(1+a^2+b^2))\\
(S^{\bf 2}(f^{a,b}_1),S^{\bf 2}(f^{a,b}_2))&=&(1/3,1/(1+a^2+b^2)).\\
\eeq
So that, the ratios $$\frac{S^{\bf 1}(f^{a,b}_i)}{S^{\bf 2}(f^{a,b}_i)},\; i=1,2$$ do not depend on $b$. Moreover, for $|a|>1$, as this ratio is greater than $1$,  $X_1$ seems to have more    influence on the output. 
\noi\\
Now let us perform a sensitivity analysis on  $\|Y\|^2$. Straightforward calculus lead to 
$$S^{\bf 1}(\|Y\|^2)\geq S^{\bf 2}(\|Y\|^2) \iff (a-1)(a^3+a^2+5a+5-4b)\geq 0.$$ 
\\
\begin{figure}[h]
\begin{center}
\includegraphics[height=5cm,width=5cm]{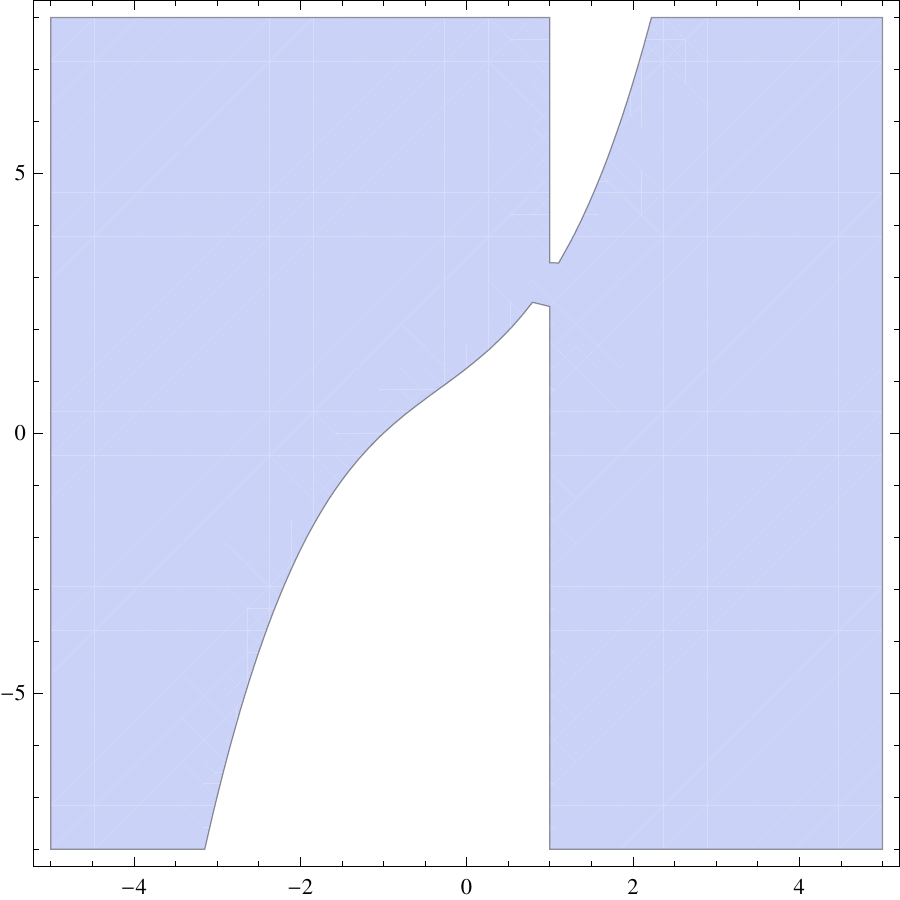}
\caption{Plot of $(a-1)(a^3+a^2+5a+5-4b)\geq 0$. The blue corresponds to regions where $S^{\bf 1}(\|Y\|^2)\geq S^{\bf 2}(\|Y\|^2)$.}
\label{plot_ex1}
\end{center}
\end{figure}

\noi \\
For the quantity $\|Y\|^2$, the region where $X_1$ is the most influent variable  depends on the value of $b$. This region is  not very  intuitive. We plot in Figure \ref{plot_ex1} the region where   $X_1$ is the most influent variable.
\end{ex}

\begin{ex}\label{ex3}
Here, we study the following two-dimensional model 

$$Y=f(X_{1},X_{2})=\begin{pmatrix}
X_1 \cos X_2\\
X_1 \sin X_2
\end{pmatrix}$$ 
with $(X_1, X_2)\sim$ Unif$([0;10])$ $\otimes$ Unif$([0; \pi/2])$.

\noi\\
We obviously  get
\beq
S^{\bf 1}(f^{a,b}_1)&=&S^{\bf 1}(f^{a,b}_2)=\frac{10}{5\pi^2-30}\approx0.52\\
S^{\bf 2}(f^{a,b}_1)&=&S^{\bf 2}(f^{a,b}_2)=\frac{3(\pi^2-8)}{4(\pi^2-6)}\approx0.36.\\
\eeq
So that $X_1$ seems to have  more influence on the output than $X_2$.

\noi\\
If we consider $\|Y\|^2$, we straightforwardly get $\|Y\|^2=X_1^2$ that does not depend on $X_2$.
\end{ex}

A last motivation to introduce new Sobol  indices is related to the statistical problem of their estimation.
As  the dimension increases  the statistical estimation of the whole  vector of scalar Sobol indices becomes more and more expensive. Moreover, the interpretation of such a large vector is not easy. This strengthens the fact that one needs  to introduce Sobol indices of small dimension, which condense all the information contained in a large collection of scalars.

In the next section we define new Sobol indices  generalizing   the scalar ones and resuming all the information.

\section{Generalized Sobol indices}\label{sec:definition}

\subsection{Definition of the new indices}

We denote by $X:=(X_1, \ldots, X_p)$ the random input,  defined on some probability space $(\Omega,{\mathcal{F}},\mathbb P)$ and valued in some measurable space $E=E_1\times\ldots \times E_p$. We denote also by $Y$ the output 
\[ Y = f(X_1, \ldots, X_p), \]
where $f: E \rightarrow \R^k$ is an unknown measurable function ($p$ and $k$ are positive integers). We assume that $X_1,\ldots,X_p$ are independent and that $Y $ is square integrable (i.e. $\E\left(\| Y\|^{2}\right)<\infty$). We also assume, without loss of generality,  that the covariance matrix of $Y$ is positive definite.

Let $\bf  u$ be a subset of $\{1, \ldots, p\}$ and denote by $\sim\!\bf u$ its complementary in $\{1, \ldots, p\}$.

Further, we set $X_{\bf u}=(X_i, i \in \bf u)$ and $E_{\bf u}=\prod_{ i \in \bf u}E_i$.\\

As the inputs $X_1, \ldots, X_p$ are independent, $f$ may be decomposed through the so-called  Hoeffding decomposition \cite{van2000asymptotic}
\begin{equation}
\label{e:hoeff}
f(X) = c + f_{\bf u}(\Xu) + f_{\sim\bf u}(\Xtu) + f_{\bf u,\sim\bf u}(\Xu,\Xtu), 
\end{equation}
where $c\in\R^k$, $f_{\bf u}: E_{\bf u}\rightarrow\R^k$, $f_{\sim\bf u}: E_{\sim\bf u}\rightarrow\R^k$ and $f_{\bf u,\sim\bf u}: E\rightarrow\R^k$ are given by
\[ c = \E(Y), \; f_{\bf u}=\E(Y|\Xu)-c, \; f_{\sim\bf u}=\E(Y|\Xtu)-c, \; f_{u,\sim\bf u}=Y-f_{\bf u}-f_{\sim\bf u}-c. \]
Thanks to $L^2$-orthogonality, computing  the covariance matrix of both sides of \eqref{e:hoeff} leads to  
\begin{equation}
\label{e:hoeffvar}
\Sigma = C_{\bf u} + C_{\sim\bf u} + C_{\bf u,\sim\bf u}.
\end{equation}
Here $\Sigma$, $C_{\bf u}$, $C_{\sim\bf u}$ and $C_{\bf u,\sim\bf u}$ are denoting respectively the covariance matrices of $Y$, $f_{\bf u}(\Xu)$, $f_{\sim\bf u}(\Xtu)$ and $f_{\bf u,\sim\bf u}(\Xu,\Xtu)$.

\begin{rmk}
Notice that for scalar outputs (i.e. when $k=1$), the covariance matrices are scalar (variances). So that \eqref{e:hoeffvar} may be  interpreted as the decomposition of the total variance of $Y$.  The summands traduce the fluctuation induced by   the input factors $X_{\bf u}$ and $X_{\sim\bf u}$,  and  the interactions between them. The (univariate) Sobol index $ S^{\bf u}(f) = \Var(\E(Y|X_{\bf u}))/\Var(Y) $
is then interpreted as the sensibility of $Y$ with respect to $X_{\bf u}$. Due to non-commutativity of the matrix product, a direct generalization of this index is not straightforward.
\end{rmk}

\noi\\
In the general case ($k \geq 2$), for any square  matrix $M$ of size $k$, the equation \eqref{e:hoeffvar} can be scalarized in the following way
\[ \Tr(M\Sigma) = \Tr(M \Cu) + \Tr(M \Ctu) + \Tr(M \Cutu). \] 
 This suggests to define as soon as $\Tr(M \Sigma)\neq 0$ the  $M$-sensitivity measure of $Y$  with respect to $X_{\bf u}$ as  
 \[ S^{\bf u}(M; f) = \frac{\Tr(M \Cu)}{\Tr(M \Sigma)} .\]
Of course  we can  analogously define
$$ S^{\sim\bf u}(M; f) = \frac{\Tr(M \Ctu)}{\Tr(M \Sigma)},\;\; S^{\bf u,\sim\bf u}(M; f) = \frac{\Tr(M \Cutu)}{\Tr(M \Sigma)}.$$ 
  The following lemma is obvious.
\begin{lem} $\;$
\begin{enumerate}
\item The generalized sensitivity measures sum up to 1
\begin{equation}
\label{e:somsobol}
S^{\bf u}(M; f) + S^{\sim \bf u}(M; f) + S^{\bf u, \sim \bf u}(M; f) = 1. 
\end{equation}
\item $0 \leq S^{\bf u}(M;f) \leq 1$.
\item Left-composing $f$ by a linear operator $O$ of $\R^k$ changes the sensitivity measure accordingly to
\begin{equation}
\label{e:changesobol}
S^{\bf u}(M; Of) = \frac{\Tr(M O \Cu O^t)}{\Tr(M O \Sigma O^t)} = 
\frac{\Tr(O^t M O \Cu)}{\Tr(O^t M O \Sigma)} = 
S^{\bf u}(O^t M O; f).
\end{equation}
\item For $k=1$ and for any $M \neq 0$, we have $S^{\bf u}(M; f) = S^{\bf u}(f)$.
\end{enumerate}
\end{lem}

\subsection{The important identity case}
We now consider the special case $M = \textrm{Id}_k$ (the identity matrix of dimension $k$).  Notice that in this case  the sensitivity indices  are the same as the ones considered through principal component analysis in \cite{lamboni2011multivariate}. We set $S^{\bf u}(f)=S^{\bf u}(\textrm{Id}_k;f)$. The index $S^{\bf u}(f)$ has the following obvious properties

\begin{proposition}
\label{prop:properties}

\begin{enumerate}
\item $S^{\bf u}(f)$ is invariant by left-composition of $f$ by any isometry of $\R^k$ i.e.
\[  \textrm{  for any  square  matrix  $O$  of size $k$  s.t. } O^tO=\textrm{Id}_k, \;\; S^{\bf u}(Of)=S^{\bf u}(f); \]
\item $S^{\bf u}(f)$ is invariant by left-composition of $f$ by any nonzero scaling of i.e.
\[  \textrm{  for any }   \lambda \in \R, \;\; S^{\bf u}(\lambda f)=S^{\bf u}(f); \]
\end{enumerate}
\end{proposition}

\begin{rmk}
The properties in this proposition  are  natural requirements for a sensitivity measure. In the next section, we will show that these requirements can be fulfilled by $S^{\bf u}(M;f)$ only when $M=\lambda \textrm{Id}_k$ ( $\lambda\in\R^*$). Hence, the  \textit{canonical choice}  among indices of the form  $S^{\bf u}(M;f)$ is the sensitivity index  $S^{\bf u}(f)$.
\end{rmk}

\subsection{Identity is the only good choice}
The following proposition can be seen as a kind of reciprocal  of Proposition \ref{prop:properties}.

\begin{proposition}\label{prop:inv}
Let $M$ be a square  matrix of size $k$ such that
\begin{enumerate}
\item $M$ does not depend  neither on $f$ nor $\bf u$;
\item $M$ has full rank;
\item $S^{\bf u}(M;f)$ is invariant by left-composition of $f$ by any isometry of $\R^k$.
\end{enumerate}
Then $S^{\bf u}(M;\cdot) = S^{\bf u}(\cdot)$.
\end{proposition}

\begin{proof} We can write $M = M_{Sym} + M_{Antisym}$ where $M_{Sym}^t=M_{Sym}$ and $M_{Antisym}^t=-M_{Antisym}$. Since, for any symmetric matrix $V$, we have $\Tr(M_{Antisym} V)=0$, we deduce that $S^{\bf u}(M;f)=S^{\bf u}(M_{Sym};f)$ ($C_u$ and $\Sigma$ being symmetric matrices). Thus we  assume, without loss of generality, that $M$ is symmetric.\\

\noi
We diagonalize $M$ in an orthonormal basis: $M=P D P^t$, where $P^t P = \textrm{Id}_k$ and $D$ diagonal. We have
\[ S^{\bf u}(M;f) = \frac{\Tr(P D P^t \Cu)}{\Tr(P D P^t \Sigma)} = \frac{\Tr(D P^t \Cu P)}{\Tr(D P^t \Sigma P)} = S^{\bf u}(D; P^tf). \]
By assumption  1. and 3., $M$ can be assumed to be diagonal.\\

\noi
Now we want to show that $M = \lambda \textrm{Id}_k$ for some $\lambda\in\R^*$. Suppose, by contradiction, that $M$ has two different diagonal coefficients $\lambda_1 \neq \lambda_2$. It is clearly sufficient to consider the case $k=2$. Choose $f=\textrm{Id}_2$ (hence, $p=2$), and ${\bf u}=\{1\}$. We have $\Sigma=\textrm{Id}_2$ and $\Cu=\left(\begin{smallmatrix} 1 & 0 \\ 0 & 0 \end{smallmatrix}\right)$. Hence on one hand $S^{\bf u}(M;f)=\frac{\lambda_1}{\lambda_1+\lambda_2}$. On the other hand, let $O$ be the isometry which exchanges the two vectors of the canonical basis of $\R^2$. We have  $S^{\bf u}(M;Of)=\frac{\lambda_2}{\lambda_1+\lambda_2}$. Thus 3. is contradicted if $\lambda_1\neq\lambda_2$. The case $\lambda=0$ is forbidden by 2. Finally, it is easy to check that, for any $\lambda\in\R^*$, $S^{\bf u}(\lambda \textrm{Id}_k;\cdot)=S^{\bf u}(\textrm{Id}_k;\cdot)=S^{\bf u}(\cdot)$. 
\end{proof}

\begin{rmk}\label{rem:min_pb}(Variational formulation)
We assume here that $\E(Y)=0$, if it is not the case, one has to consider the centered variable $Y-\E(Y)$. As in dimension 1, one can see that this new index can also be seen as the solution of the following least-squares problem (see \cite{janon2012asymptotic}) 
\begin{equation*}
\argmin_{a}\E\|Y^{\bf u}-aY\|^{2}.
\end{equation*}

As a consequence, $S^{\bf u}(f)Y$ can be seen as the projection of $Y^{\bf u}$ on $\{aY, a\in \R\}$.
\end{rmk}
 
\begin{rmk}
Notice that the condition $\Tr(\Sigma)\neq 0$ (necessary for the indices to be well-defined) is fulfilled as soon as $Y$ is not constant.
\end{rmk}

 \noi
We now give two toy examples to illustrate our definition.

\begin{ex}
We consider as first example  $$Y=f^{a}(X_{1},X_{2})=\begin{pmatrix}
aX_{1}\\
X_{2}
\end{pmatrix},$$
with $X_{1}$ and $X_{2}$ i.i.d. standard Gaussian random variables. We easily get
$$S^{\bf 1}(f)=\frac{a}{a^2+1} \quad \trm{and} \quad S^{\bf 2}(f)=\frac{1}{a^2+1}=1-S^{\bf 1}(f).$$ 
\end{ex}

\begin{ex} We consider Example \ref{ex0}
\label{ex2}
$$Y=f^{a,b}(X_{1},X_{2})=\begin{pmatrix}
X_{1}+X_{1}X_{2}+X_{2}\\
aX_{1}+bX_{1}X_{2}+X_{2}
\end{pmatrix}.$$ 

\noi
We have
$$S^{\bf 1}(f)=\frac{1+a^2}{4+a^2+b^2} \quad \trm{and} \quad  
S^{\bf 2}(f)=\frac{2}{4+a^2+b^2}$$
and obviously 
$$S^{\bf 1}(f)\geq S^{\bf 2}(f) \iff a^2\geq 1.$$
This result has the natural interpretation that, as $X_1$ is scaled by $a$, it has more influence if and only if this scaling enlarges $X_1$'s support i.e. $|a|>1$.
\end{ex}

\section{About uniqueness}\label{sec:uniqueness}
In this section, we show that it is possible to build other indices having the same invariance properties as $S^{\bf u}(f)$. 

\subsection{Another index}
 
Here we use \eqref{e:hoeffvar} in a different way to get a natural definition of a Sobol matricial index with respect to the variable $X_{\bf u}$. Indeed, we may choose as Sobol matricial index

\begin{equation}\label{eq:cond_AB}
BC_{\bf u}A 
\end{equation}
for any matrices $A$ and $B$ such that $AB=\Sigma^{-1}.$ 	
First, note that this index is a square matrix of size $k$. Second, any convex combination of Sobol matricial indices of the form \eqref{eq:cond_AB} is still a good candidate for the Sobol matricial index with respect to $X_{\bf u}$.

\noi\\
\begin{rmk}(Another variational formulation)
In the spirit of Remark \ref{rem:min_pb} (with the same assumption $\E(Y)=0$), consider the following minimization problem
\begin{equation}
\argmin_{M\in \mathcal{M}_k}\E\|P^{t}Y^{\bf u}-MP^{t}Y\|^{2},
\end{equation}
where $P$ is a matrix such that $P^{t} \Sigma P$ is diagonal and $\mathcal{M}_k$ is the set of all square matrices of size $k$. The solution
of this minimization problem 
$$
P^{t}C_{\bf u} \Sigma^{-1}P
$$
is a  good candidate to be a Sobol matricial index (with $A=\Sigma^{-1}P$ and $B=P^{t}$).

\noi
Note now that the symmetric version of this Sobol matricial index 
$$
P^{t}\Sigma^{-1}C_{\bf u} P
$$
is also a good candidate (here $A=P$ and $B=P^{t}\Sigma^{-1}$).
\end{rmk}

\noi
In order to warrant  that a Sobol matricial index fulfills a consistent definition, we should require a little bit more.
First, a reasonable condition should be that the Sobol matricial index is a symmetric matrix: the influence of the input $X_i$  on the coordinates $k$ and $l$ of the output $Y$ should be the same as the influence of the input $X_i$  on the coordinates $l$ and $k$ of $Y$.\\
Secondly, the Sobol matricial index should share the  properties of  the scalar index $S^{\bf u}(f)$. That is, it should be invariant by any isometry, scaling and translation. This leads to the definition of a family of matricial indices.\\

For the sake of simplicity, we assume that the eigenvalues of $\Sigma$ are simple. Let $0<\lambda_1< \ldots < \lambda_k$ be the ordered eigenvalues and let $(O_i)_{i=1,\ldots,k}$ be  such that $O_i $ is the unit eigenvector associated to $\lambda_i$ whose first non-zero coordinate is positive. Let $O$  be the (orthogonal) matrix whose column $i$ is $O_i$.\\
Let $\mathcal{H}_{k}$ be the group of signed permutations matrix of size $k$.  That is, $P\in \mathcal{H}_{k}$ if and only if each row and each column of $P$ has exactly one non zero element, which belongs to $\{-1,1\}$. 

Notice that any orthogonal matrix that diagonalizes $\Sigma$ can be written as $OP$, where
$P \in \mathcal H_k $. Suppose that $\mu$ is a probability measure on $\mathcal H_k$. We define

\begin{align}\label{eq:T}
T^{\bf u,\mu}&=
\frac12\left(\int_{ \mathcal{H}_k}(OP)^{t} \left(\Sigma^{-1}C_{\bf u}+C_{\bf u}\Sigma^{-1}\right)OP\mu(dP)\right).
\end{align}

We then have the following Proposition.

\begin{proposition}
$T^{\bf u,\mu}$ is invariant by any isometry, scaling and translation.
\end{proposition}

\begin{proof}
Let $U$ be an isometry of $\R^{k}$ and set
 $$
W=UY, \ W^{\bf u}= UY^{\bf u}.
$$
It is clear that  $\Sigma_W=U\Sigma U^{t}$ and $C_{\bf{u},W}=\Cov(W,W^{\bf u})=U C_{\bf u} U^t$. Since $O$ diagonalizes $\Sigma$,  $O_W=UO$ diagonalizes the covariance matrix $\Sigma_W$ of $W$.  Then, for any $P\in \mathcal{H}_k$,
\begin{align*}
(O_WP)^{t} \Sigma_{W}^{-1}C_{\bf{u},W}(O_WP)&=P^tO^{t} U^{t}U\Sigma_{Y}^{-1}U^{t}UC_{\bf u}U^{t}UOP=(OP)^{t} \Sigma_{Y}^{-1}C_{\bf u}OP\\
(O_WP)^{t}C_{\bf u,W} \Sigma_{W}^{-1}(O_WP)&=P^{t}O^t U^{t}U C_{\bf u}  U^{t}U\Sigma_{Y}^{-1}U^{t}UOP=(OP)^{t} C_{\bf u}\Sigma_{Y}^{-1}OP
\end{align*}
By integrating the above equalities with respect to $\mu$, we obtain the invariance by isometry. The other invariances are obvious.
\end{proof}

\begin{rmk}
At first look, one may also consider matricial indices based on 
$$\Sigma^{-\alpha}C_{\bf u}\Sigma^{-\beta}+\Sigma^{-\beta}C_{\bf u}\Sigma^{-\alpha}, \quad \trm{with} \quad \alpha+\beta=1.$$
Nevertheless, these matricial indices are not admissible even if $\Sigma^{-\alpha}\Sigma^{-\beta}=\Sigma^{-1}$ (see \eqref{eq:cond_AB}) since they are not invariant by isometry.\end{rmk}

\begin{rmk}
For the sake of simplicity, we have restricted ourselves to the generic case where all eigenvalues of $\Sigma$ are simple. When there is only $l$ ($l<k$) distinct eigenvalues, the group $\mathcal{H}_k$  has to be replaced by the much more complicated set of all isomorphisms $P$ on 
$\R^k$ that can be written as
$$P=\Pi O_1\ldots O_l,$$
where $\Pi$ is some permutation on $\{1,\ldots,k\}$ and $O_i$ ($i=1,\ldots, l$) is some isometry on $\R^k$ letting invariant the orthogonal of the eigenspace  associated with the $i^{\text{th}}$ eigenvalue. 
\end{rmk}
Let $\mu^\star$ be the uniform probability measure on the finite set $\mathcal H_k$.

\begin{lem}
Let $A$ be a square matrix of size $k$. Then 
$$
\int_{\mathcal H_{k}}P^{t} A P\mu^\star(dP)=\frac{\Tr(A)}{k}I_k.
$$
\end{lem}

\begin{proof}
One can see that $\mathcal H_{k}=\left\{D_{\epsilon}P_{\sigma}; \epsilon \in \{-1,1\}^k \; \trm{and} \; \sigma \; \trm{a permutation of} \; \{1,\ldots,k\} \right\}$ where $D_{\epsilon}=\mathop{\mathrm{diag}}(\epsilon)$ and $P_{\sigma}$ the permutation matrix associated to $\sigma$ that is $\left(P_\sigma\right)_{i,j}=\left\{ \begin{array}{l} 1 \text{ if } j=\sigma(i) \\ 0 \text{ else.} \end{array} \right.$\\ 

Set $P_{\epsilon,\sigma}=D_{\epsilon}P_{\sigma}$ the element  of $\mathcal{H}_k$ associated to    $\epsilon$ and $\sigma$.

Then
\beq
B&:=&\int_{\mathcal P_{k}}P_{\epsilon,\sigma}^{t} A P_{\epsilon,\sigma}\mu^\star(dP_{\epsilon,\sigma})=\frac{1}{2^k k!}\sum_{\sigma}P_{\sigma}^{t}\left(\sum_{\epsilon}  D_{\epsilon}^t A D_{\epsilon}\right)P_{\sigma}\\
\eeq
We have $\left(D_{\epsilon}^t A D_{\epsilon}\right)_{ij}=\epsilon_i A_{ij} \epsilon_j$, hence 
\beq
\left(\sum_{\epsilon}  D_{\epsilon}^t A D_{\epsilon}\right)_{ij}&=&A_{ij}\sum_{\epsilon}  \epsilon_i \epsilon_j
=\begin{cases}
0 & \trm{if} \ i\neq j\\
2^kA_{ii} & \trm{if} \ i=j.\\
\end{cases}
\eeq 
Thus
\begin{eqnarray*}
B_{ij}&=&\frac{1}{k!} \sum_{\sigma} \sum_{l=1}^k (P_{\sigma})_{li} \left( \mathop{\mathrm{diag}}(A_{11},\ldots,A_{kk})P_{\sigma} \right)_{lj} \\
&=& \frac{1}{k!} \sum_\sigma (P_\sigma)_{\sigma^{-1}(i),i} \left( \mathop{\mathrm{diag}}(A_{11},\ldots,A_{kk})P_{\sigma} \right)_{\sigma^{-1}(i),j} \\
&=& \begin{cases} 0 & \trm{ if } i\neq j \\
               \frac{1}{k!} \sum_\sigma A_{\sigma^{-1}(i),\sigma^{-1}(i)}(i) & \trm{ if } i=j \end{cases}
\end{eqnarray*}
and we have
\[ B_{ii} = \frac{1}{k!} \sum_{l=1}^k \left( \sum_{\sigma\trm{ s.t. } \sigma^{-1}(i)=l} A_{ll} \right) = \frac{1}{k!} \sum_{l=1}^k (k-1)! A_{ll} = \frac{\Tr(A)}{k}. \;\;\;\;\;\qedhere \]
\end{proof}

\noi
Using the previous Lemma in conjunction with \eqref{eq:T}, we obtain the following Sobol matricial index 
\begin{equation} 
T^{\bf u}:=T^{\bf u,\mu^{\star}}= \frac{\Tr\left(\Sigma^{-1}C_{\bf u}\right)}{k}I_k. 
 \end{equation}
Notice that this matricial index only depends on the real number $\Tr\left(\Sigma^{-1}C_{\bf u}\right)/\Tr(I_k)$ which is easy to interpret. \\

\subsection{Comparison between $S^{\bf u}(f)$ and $T^{\bf u}$}

We have defined two nice candidates to generalize the scalar Sobol index in dimension $k$. A natural question is:
which one should be preferred? There is \emph{a priori} no universal answer.\\ 
Nevertheless, from a statistical point of view, $T^{\bf u}$ presents a major drawback: its estimation may require the estimation of an inverse covariance matrix $\Sigma^{-1}$, which may be tricky. While the estimation of $S^{\bf u}(f)$ only uses estimation of traces of covariance matrices. Besides, the following example  shows that $T^{\bf u}$ may be useless in some models.

\begin{ex}
We consider again the model of Example \ref{ex2}
$$Y=f^{a,b}(X_{1},X_{2})=\begin{pmatrix}
X_{1}+X_{1}X_{2}+X_{2}\\
aX_{1}+bX_{1}X_{2}+X_{2}
\end{pmatrix}.$$ 
We easily get

$$T^{ \bf 1}=\frac{(b-a)^2+(a-1)^2}{4[(b-a)^2+(a-1)(b-1)]}I_2,
\;\; \; \; \; T^{\bf 2}=\frac{(b-1)^2+(a-1)^2}{4[(b-a)^2+(a-1)(b-1)]}I_2. $$

\noi Thus
$$T^{ \bf 1}\geq T^{ \bf 2} \iff (a-1)(a-2b+1)\geq 0$$
whereas we have obtained previously, the more intuitive result
$$S^{\bf 1}(f)\geq S^{\bf 2}(f) \iff a^2\geq 1.$$

Moreover $T^{\bf u}$ is not informative since for $a=1$, the indices $T^{\bf 1}$ and $T^{\bf 2}$ satisfy
$$T^{ \bf 1}=T^{\bf 2}=\frac{1}{4}I_2$$
and do not depend on $b$.
\end{ex}

Thus, it seems to us that $S^{\bf u}(f)$ is a more relevant sensitivity measure, and, in the sequel, we will focus our study on $S^{\bf u}(f)$. \\

\section{Estimation of $S^{\bf u}(f)$}\label{sec:estimation}

\subsection{The Pick and Freeze estimator}
 In practice, the covariance matrices $\Cu$ and $\Sigma$ are not analytically available. In the scalar case ($k=1$), it is customary to estimate $ S^{\bf u}(f) $ by using a Monte-Carlo Pick and Freeze method \cite{sobol1993,janon2012asymptotic}, which uses a finite sample of evaluations of $f$.\\\\
  In this Section, we propose a Pick and Freeze estimator for the vectorial case which generalizes the $T_N$ estimator studied in \cite{janon2012asymptotic}. We set $Y^{\bf u}=f(X_{\bf u}, X_{\sim\bf u}')$ where $X_{\sim\bf u}'$ is an independent copy of $X_{\sim\bf u}$ which is still independent of $X_{\bf u}$. Let $N$ be an integer. We take $N$ independent copies $Y_1, \ldots, Y_N$ (resp. $Y_1^{\bf u},\ldots,Y_N^{\bf u}$) of $Y$ (resp. $Y^{\bf u}$). For $l=1,\ldots,k$, and $i=1,\ldots,N$, we also  denote by $Y_{i, l}$ (resp. $Y_{i, l}^{\bf u}$) the $l^\text{th}$ component of $Y_i$ (resp. $Y_i^{\bf u}$). We then define the following estimator of $S^{\bf u}(f)$
\begin{equation}\label{def:SN}
 S_{{\bf u}, N} = \frac{\sum_{l=1}^k \left( \frac 1 N \sum_{i=1}^N Y_{i,l} Y_{i,l}^{\bf u} - \left( \frac1N  \sum_{i=1}^N \frac{Y_{i,l}+Y_{i,l}^{\bf u}}{2} \right)^2 \right)    }{ \sum_{l=1}^k \left( \frac 1 N \sum_{i=1}^N \frac{ Y_{i,l}^2+(Y_{i,l}^{\bf u})^2 }{2} - \left( \frac1N \sum_{i=1}^N \frac{ Y_{i,l}+Y_{i,l}^{\bf u} }{2} \right)^2  \right) }. 
 \end{equation}

\begin{rmk}\label{rem:SN}
Note that this estimator can be written
 \begin{equation} 
 S_{{\bf u}, N}= \frac{\Tr\left(C_{{\bf u},N}\right)}{\Tr\left(\Sigma_N\right)}
 \end{equation}
where $C_{{\bf u},N}$ and $\Sigma_N$ are the  empirical estimators of $C_{\bf u}=\Cov(Y,Y^{\bf u})$ and $\Sigma=\Var(Y)$  defined by
$$C_{{\bf u},N}=\frac 1 N \sum_{i=1}^N Y_i^{\bf u}Y_i^t- \left(\frac 1 N \sum_{i=1}^N 
\frac{Y_{i}+Y_{i}^{\bf u}}{2}\right)\left(\frac 1 N \sum_{i=1}^N \frac{Y_{i}+Y_{i}^{\bf u}}{2}\right)^t$$
and
$$\Sigma_N=\frac 1 N \sum_{i=1}^N  \frac{ Y_{i}Y_{i}^t+Y_{i}^{\bf u}(Y_{i}^{\bf u})^t }{2} -\left(\frac 1 N \sum_{i=1}^N 
\frac{Y_{i}+Y_{i}^{\bf u}}{2}\right)\left(\frac 1 N \sum_{i=1}^N \frac{Y_{i}+Y_{i}^{\bf u}}{2}\right)^t.$$
\end{rmk}

\subsection{Asymptotic properties}

\noi
A straightforward application of the Strong Law of Large Numbers leads to

\begin{proposition}[Consistency]
 $S_{{\bf u}, N}$ converges almost surely to  $S^{\bf u}(f)$ when $N \rightarrow +\infty$.
\end{proposition}

We now study to the asymptotic normality of $(S_{{\bf u},N})_N$.

\begin{proposition}[Asymptotic normality]
\label{prop:an}
Assume $\E(Y_{l}^4)<\infty$ for all $l=1,\ldots,k$. For $l=1,\ldots,k$, we set
\[ U_l = (Y_{1,l}-\E(Y_l))(Y_{1,l}^{\bf u}-\E(Y_l)), \qquad
   V_l = ( Y_{1,l}-\E(Y_l) )^2 + ( Y_{1,l}^{\bf u}-\E(Y_l) )^2. \]

Then
\begin{equation}\label{clt}
\sqrt{N} \left( S_{{\bf u}, N}-S^{\bf u}(f) \right)
\overset{\mathcal{L}}{\underset{N\to\infty}{\rightarrow}}
\mathcal{N}_1\left(0,\sigma^2\right)
\end{equation}
where
\begin{equation}\label{sigma2} \sigma^2 = a^2 \sum_{l,l' \in \{1,\ldots,k\}} \Cov(U_l, U_{l'}) + b^2 \sum_{l,l' \in \{1,\ldots,k\}} \Cov(V_l, V_{l'}) + 2 a b \sum_{l,l' \in \{1,\ldots,k\}} \Cov(U_l, V_{l'}), \end{equation}
with \[ a = \frac{1}{\sum_{l=1}^k \Var(Y_l)}, \qquad b=-\frac{a}{2} S^{\bf u}(f). \]
\end{proposition}

\begin{proof}
Since $S_{{\bf u},N}$ remains invariant when $Y$ is changed to $Y - \E(Y)$, we have \[ S_{{\bf u}, N} = \Phi\left( \frac{1}{N} \sum_{i=1}^N W_i \right), \] where
\[
W_i = \begin{pmatrix} ( Y_{i1}-\E(Y_{1}) )(Y_{i1}^{\bf u}-\E(Y_1)) \\
\vdots \\
 Y_{ik}-\E(Y_{k}) )(Y_{ik}^{\bf u}-\E(Y_k)) \\
 Y_{i1}-\E(Y_1) + Y_{i1}^{\bf u} - \E(Y_1) \\
 \vdots \\
 Y_{ik}-\E(Y_k) + Y_{ik}^{\bf u} - \E(Y_k) \\
 ( Y_{i1}-\E(Y_1) )^2 + ( Y_{i1}^{\bf u}-\E(Y_1) )^2 \\
 \vdots \\
 ( Y_{ik}-\E(Y_k) )^2 + ( Y_{ik}^{\bf u}-\E(Y_k) )^2
 \end{pmatrix} \]
and
\[ \Phi(x_1,\ldots,x_k,y_1,\ldots,y_k,z_1,\ldots,z_k) = 
\frac{ \sum_{l=1}^k \left( x_l - (y_l/2)^2 \right) }{\sum_{l=1}^k \left( z_l/2 - (y_l/2)^2 \right) }. \]

The so-called Delta method (\cite{van2000asymptotic}, Theorem 3.1) gives
\[ \sqrt N (S_{{\bf u},N} - S^{\bf u}(f)) 
\overset{\mathcal{L}}{\underset{N\to\infty}{\rightarrow}}
\mathcal{N}_1\left(0,\sigma^2\right) \]
where $\sigma^2 = g^t \Gamma g$, $\Gamma$ the covariance matrix of $W_1$ and
\[ g = \nabla\Phi(\E(W_1)). \]
We have
\[ \E(W_1) = \left( \Cov(Y_1,Y_1^{\bf u}), \ldots, \Cov(Y_k,Y_k^{\bf u}), 0, \ldots, 0, 2 \Var Y_1, \ldots, 2 \Var Y_k \right)^t, \]
and by differentiation of $\Phi$, $g = \left( a, \ldots, a, 0, \ldots, 0, b, \ldots, b \right)^t.$ A simple matrix calculus leads to \eqref{sigma2}.
\end{proof}

\begin{rmk}
Following the same idea, it is possible, for $\bf v \subset \{1,\ldots,p\}$, to derive a (multivariate) central limit theorem for $$\left( S_{{\bf u},N},S_{{\bf v},N}, S_{{\bf u \cup v},N}\right)=\left(\frac{\Tr\left(C_{{\bf u},N}\right)}{\Tr\left(\Sigma_N\right)},\frac{\Tr\left(C_{{\bf v},N}\right)}{\Tr\left(\Sigma_N\right)},\frac{\Tr\left(C_{{\bf u \cup v},N}\right)}{\Tr\left(\Sigma_N\right)}\right).$$
We then can derive a (scalar) central limit theorem for $ S_{{\bf u\cup v}, N} - S_{{\bf u},N} - S_{{\bf v},N} $, a natural estimator of $S^{\bf u \cup v} - S^{\bf u} - S^{\bf v}$, which quantifies the influence  (for $u \cap v=\emptyset$) of the interaction between the variables of ${\bf u}$ and ${\bf v}$.
\end{rmk}

\begin{proposition}\label{pro:dimfinie}
Assume $\E(Y_{l}^4)<\infty$ for $l=1,\ldots,k$. Then
$\left(S_{{\bf u},N}\right)_{N}$ is asymptotically efficient for estimating $S^{\bf u}(f)$ among regular estimator sequences that are function of the exchangeable pair $(Y, Y^{\bf u})$.
\end{proposition}

\begin{proof}
Note that
$$
S^{\bf u}(f) = \Phi\left(\Tr\left(C_{{\bf u}}\right),\Tr\left(\Sigma_{Y}\right)\right)
 \;\;\;\; \mbox{and }S_{{\bf u},N}= \Phi\left(\Tr\left(C_{{\bf u},N}\right),\Tr\left(\Sigma_N\right)\right)$$
where $\Phi$ is defined by $\Phi(x,y)={x}/{y}$.\\
Proceeding as in the proof of Proposition 2.5 in \cite{janon2012asymptotic}, we derive that $C_{{\bf u},N}$ (respectively $\Sigma_N$) is asymptotically efficient for estimating $C_{{\bf u}}$ (resp. $\Sigma_{Y}$). 
Then, Theorem 25.50 (efficiency in product space) in \cite{van2000asymptotic} gives that
$\left(C_{{\bf u},N},\Sigma_N\right)$ is asymptotically efficient for estimating $\left(C_{{\bf u}},\Sigma_{Y}\right)$. \\
Now since $\Phi$ (respectively $\Tr$) is   differentiable in $\R^2 \setminus \{y=0\}$ (resp. in $\mathcal{M}_{k}$) we can apply Theorem 25.47 in \cite{van2000asymptotic}  (efficiency and Delta method) to get that $\left(\Phi\left(\Tr\left(C_{{\bf u},N}\right),\Tr\left(\Sigma_N\right)\right)\right)_{N}$ is also asymptotically efficient for estimating $\Phi\left(\Tr\left(C_{{\bf u}}\right),\Tr\left(\Sigma_{Y}\right)\right)$.
\end{proof}

\subsection{Concentration inequality}\label{s:concentre}

In this section we apply Corollary 1.17 of Ledoux \cite{Ledoux01} to give a concentration inequality for $S_{{\bf u},N}$. In order to be self-contained we recall Ledoux's result.

\begin{cor}[Corollary 1.17 of \cite{Ledoux01}]
Let $P=\mu_1\otimes\ldots\otimes\mu_n$ be a product probability measure on the cartesian product $X=X_1\times\ldots\times X_n$ of metric spaces $(X_i,d_i)$ with finite diameters $D_i$, $i=1,\ldots, n$, endowed with the $l^1$ metric$d=\sum_{i=1}^n d_i$. Let $F$ be a 1-Lipschitz function on $(X,d)$. Then, for every $r\geq 0$, 
$$P\left(F\geq \E_P(F)+r \right)\leq e^{-r^2/2D^2}$$
where $D^2= \sum_{i=1}^n D_i^2$.
\end{cor}
Our concentration inequality is the following
\begin{proposition} \label{prop_concentr}
Assume that $Y$ is bounded almost surely in $\R^k$, that is, there exists $\rho>0$ so that $\norm{Y}_2 < \rho$.
Let $v_l:=\Sigma_{l,l} / \rho^2$ for $l=1\ldots k$. \medskip

Then, for all $t \geq 0$, we have
\[ \P \left( S_{{\bf u},N}-S^{\bf u}(f) \geq t\right) \leq \exp\left(- \frac{N}{32} \left( \frac{ t - \frac{1}{2N} (\suf+t-1)(\suf+1) }{{1+\suf+t}+|\suf+t-1|} \sum_{l=1}^k v_l \right)^2  \right), \]
and, for all $t \geq \frac{ (1-\suf)(1+\suf) }{ 2N - (1+\suf)}$, we have
\[ \P \left( S_{{\bf u},N}-S^{\bf u}(f) \leq -t\right) \leq \exp\left(- \frac{N}{32} \left( \frac{  t + \frac{1}{2N} (\suf-t-1)(\suf+1) }{{1+\suf-t}+|\suf-t-1|} \sum_{l=1}^k v_l \right)^2  \right)  . \] 
\end{proposition} 

\begin{rmk}
Note that the use of Corollary 1.17 of Ledoux \cite{Ledoux01} leads to  bounds  improving the one obtained in \cite{pickfreeze}.
\end{rmk}

\begin{proof}
Since $S^{\bf u}(f)$ and $S_{{\bf u},N}$ are invariant by homothety, one can scale the output $Y$ so that $Y\in B_k(0,1)$, the unit Euclidean ball of $\R^k$. From now on, we assume that $Y \in B_k(0,1)$ and $\rho=1$.

By Remark \ref{rem:SN}, one gets
\begin{equation}
\label{e:superastuce}
\P\left(S_{{\bf u},N}-S^{\bf u}(f)\geq t\right)=
\P\left(\Tr \left(C_{{\bf u},N}\right)-(S^{\bf u}(f)+t)\Tr \left(\Sigma_N\right)\geq 0\right). 
\end{equation}

\noi
Now let $G: \left(B_k(0,1)\times B_k(0,1)\right)^N \rightarrow \R$ defined by
\begin{multline*}
\!\!G((x_1,y_1),\ldots,(x_N,y_N))=\sum_{l=1}^k \Biggl[\frac{1}{N}\sum_{i=1}^N \left(x_{i,l}y_{i,l}-(S^{\bf u}(f)+t)\frac{(x_{i,l})^2+(y_{i,l})^2}{2}\right)^{}\\
+ (S^{\bf u}(f)+t-1)\left(\frac{1}{N}\sum_{i=1}^N \frac{x_{i,l}+y_{i,l}}{2}\right)^2\Biggr],
\end{multline*}
with $x_i=(x_{i,l})_{l=1,\ldots,k}$ and $y_i=(y_{i,l})_{l=1,\ldots,k}$ for all $i=1,\ldots,N.$

A simple computation gives that
\[ G\left(\left(Y_1,Y_1^{\bf u}\right),\ldots,\left(Y_N,Y_N^{\bf u}\right)\right) = \Tr \left(C_{{\bf u},N}\right)-(S^{\bf u}(f)+t)\Tr \left(\Sigma_N\right). \]
We have:
\[ \frac{\partial G}{\partial x_{i}}=\left( \frac{\partial G}{\partial x_{i,l}} \right)_{l=1,\ldots,k}=   \frac 1 N \left( y_i - \left(\suf +t\right) x_i + \left(\suf+t-1\right) w \right) \]
and symmetrically
\[ \frac{\partial G}{\partial y_{i}}=\left( \frac{\partial G}{\partial y_{i,l}} \right)_{l=1,\ldots,k}=   \frac 1 N \left( x_i - \left(\suf +t\right) y_i + \left(\suf+t-1\right) w \right), \]
where
\[
w = \frac 1 N \sum_{r=1}^N \frac{ x_r + y_r }{2}.
\]
Applying several times the triangular inequality and that $\left\|w\right\|_2 \leq 1 $, we deduce
\[
\left\| \frac{\partial G}{\partial x_i} \right\|_2
 \leq \frac 1 N (|1+\suf+t|+|\suf+t-1|)
\]
and
\[
\left\| \frac{\partial G}{\partial y_i} \right\|_2
 \leq \frac 1 N (|1+\suf+t|+|\suf+t-1|).
\]
Thus, $G$ is $L$-Lipschitz with $L:=\frac 1N (1+\suf+t+|\suf+t-1|)$.\\

\noi
Now we apply Corollary 1.17 of Ledoux \cite{Ledoux01} with

\begin{itemize}
\item[$\bullet$] $X_i=B_k(0,1)\times B_k(0,1)$ endowed with the metric $d_i$ defined by
$$d_i(z,z'):=\|x-x'\|_2+\|y-y'\|_2$$ 
for $z=(x,y)\in X_i$ and $z'=(x',y')\in X_i$, $x$, $x'$, $y$ and $y' \in B_k(0,1)$,
\item[$\bullet$] $X=X_1\times \ldots \times X_N$, with the $l^1$-metric $d=\sum_{i=1}^N d_i$,
\item[$\bullet$] $D_i=diam(X_i)=2+2=4$ and $D^2=\sum_{i=1}^N D_i^2=16N$,
\item[$\bullet$] $F=G/L$,
\item[$\bullet$] $r=-\E(F)=\left[t-\frac{1}{2N}(S^{\bf u}(f)+t-1)(S^{\bf u}(f)+1)\right]\sum_{l=1}^k v_l/L \geq 0$ as $0\leq\suf\leq 1$.
\end{itemize}
We then get the upper deviation bound of \eqref{e:superastuce}.

To get the second bound, we repeat the procedure by replacing $G$ (respectively $t$) by $-G$ (resp. $-t$). Note that in this case, we take
\[ r = - \left( -t - \frac{1}{2N}(\suf-t-1)(\suf+1) \right)\sum_{l=1}^k v_l/L, \]
which is non-negative thanks to the minoration hypothesis on $t$.
\end{proof}

\noi
The bounds in Proposition \ref{prop_concentr} depend on the unknown quantity $S^{\bf u}(f)$ which can not be computed in practice. To address this problem, we use the bound $0 \leq S^{\bf u}(f) \leq 1$ to get:

\begin{cor}
\label{cor:concentr_pire}
Let $V=\left(\sum_{l=1}^k v_l \right)^2$. We have 
\begin{align}
\label{eq:concentr_pire1}
&\forall t \geq0,\;\; \P\left(S_{{\bf u},N}-S^{\bf u}(f)\geq t\right)  \leq \exp\left(- \frac{NV}{128} \left(1-\frac1N\right)^{2} \left(\frac{t}{1+t}\right)^2  \right), \\
\label{eq:concentr_pire2}
&\forall t \in \left] \frac{9}{8N} , 1 \right[, \;\; \P\left(S_{{\bf u},N}-S^{\bf u}(f)\leq -t\right)  \leq \exp\left( - \frac{NV}{128}\left(t - \frac{9}{8N} \right)^2 \right).
\end{align}
\end{cor} 

\begin{proof}
\begin{enumerate}
\item Proof of \eqref{eq:concentr_pire1}: by Proposition \ref{prop_concentr} we have
\[ \P \left( S_{{\bf u},N}-S^{\bf u}(f) \geq t\right) \leq \exp\left(- \frac{N V H(t,\suf)}{32}  \right), \]
with 
\[ H(t,\suf)=\left( \frac{ t - \frac{1}{2N} (\suf+t-1)(\suf+1) }{{1+\suf+t}+|\suf+t-1|} \right)^2. \]
If $0\leq \suf\leq 1-t$ we have
\[
H(t,\suf)=\frac 14\left(  t - \frac{1}{2N} (\suf+t-1)(\suf+1)  \right)^2,
\]
and so
\[ H(t,\suf) \geq \frac{t^2}{4}\geq\frac 1 4 \left(\frac{t}{(1+t)}\right)^{2}\left(1-\frac 1N\right)^{2}, \]
as, in this case,
\[ \frac{1}{2N}(\suf+t-1)(\suf+1) \leq 0. \]
Now if  $\suf\geq 1-t,$  
$$
H(t,\suf)=\frac14\left( \frac{ t - \frac{1}{2N} (\suf+t-1)(\suf+1) }{\suf+t} \right)^2.
$$
and we have
\[ \frac{ t - \frac{1}{2N} ( \suf + t - 1 )( \suf + 1 ) }{\suf+t} \geq \frac{ t - \frac t N}{1+t} \geq 0, \]
hence in this case
$$
H(t,\suf)\geq\frac14\left(\frac{t}{1+t}\right)^{2}\left(1-\frac 1N\right)^{2}.
$$
Finally, in all cases, we have \eqref{eq:concentr_pire1}.
\item  Proof of \eqref{eq:concentr_pire2}: hypothesis $t\geq 9/(8N)$ ensures that the second part of Proposition \ref{prop_concentr} can be applied. We have
\[ \P \left( S_{{\bf u},N}-S^{\bf u}(f) \leq -t\right) \leq \exp\left(- \frac{N V H(t,\suf)}{32}  \right), \]
with 
\[ H(t,\suf)=\left( \frac{ t + \frac{1}{2N} (\suf-t-1)(\suf+1) }{{1+\suf-t}+|\suf-t-1|} \right)^2. \]
As necessarily $\suf-t-1 \leq 0$, we have:
\[ H(t,\suf) = \frac{1}{4} \left( t+\frac{1}{2N}(\suf-t-1)(\suf+1)\right)^2, \]
and, as $(\suf-t-1)(\suf+1)$ is minimal when $\suf=t/2$, we have
\[ t+\frac{1}{2N} ( \suf - t - 1 )(\suf+1) \geq t - \frac{1}{2N} \left(\frac{t}{2}+1\right)^2 \geq t - \frac{9}{8N},\]
as $t \leq 1$. Hence
\[ H(t,\suf) \geq \frac{1}{4} \left( t-\frac{9}{8N} \right)^2, \]
which completes the proof. \qedhere

\end{enumerate}
\end{proof}


\subsection{Numerical illustrations}
\label{sec:numill}

In this section, we provide numerical simulations for the sensitivity indices $S^{\bf u}(f)$ defined in Section \ref{sec:definition}.

\subsubsection{Toy example}
We consider again Example \ref{ex2} with $k=p=2$, $a=2$ and $b=3$ which leads to the following model
\[ Y = f(X_1, X_2) = \begin{pmatrix} X_1 + X_2 + X_1 X_2 \\ 2 X_1 + 3 X_1 X_2 + X_2 \end{pmatrix}. \]

In the ``Gaussian case'' (respectively ``Uniform case''), we take $X_1$ and $X_2$ independent standard Gaussian random variables (resp. independent uniform random variables on $[0;1]$).
In these two cases, a simple analytic calculus yields the true values of the sensitivity indices $S^1(f)$ and $S^2(f)$.

\paragraph*{Asymptotic confidence interval}

We perform 100 simulations of the estimated Pick and Freeze confidence interval given by Proposition \ref{prop:an} for $N=100, 200$ and $10 000$. In each case, we estimate the coverage of the 95\% confidence interval procedure by counting the proportion of estimated intervals containing the true value.

The results are gathered in Table \ref{t:1}. We see that the estimated coverages are close to the theoretical level of $95\%$ with a coverage higher than the theoretical one in the Uniform case and lesser in the Gaussian case.

%

\begin{table}
\centering
\begin{tabular}{|c|c|c|c|c||c|}
\cline{3-6}
\multicolumn{2}{c|}{} & $N$=100 & $N$=2000 & $N$=10000 & True value \\ \hline
\multirow{2}{*}{Gaussian case} & $S^1(f)$ & 0.97 & 0.94 & 0.97 & 0.2941 \\
                               & $S^2(f)$ & 0.94 & 0.93 & 0.93 & 0.1176 \\ \hline
\multirow{2}{*}{Uniform case} & $S^1(f)$ & 1 & 1 & 1 & 0.6084 \\
                               & $S^2(f)$ & 0.97 & 0.98 & 0.97 & 0.3566 \\ \hline

\end{tabular}
\caption{Estimated coverages of the 95\% confidence intervals for $S^1(f)$ and $S^2(f)$. }
\label{t:1}
\end{table}

\paragraph*{Concentration inequality}
We notice that the concentration inequality (Proposition \ref{prop_concentr}) can not be applied to the Gaussian case since $\norm{Y}_2$ is not bounded. Hence we only study the Uniform case.

For different values of $t$, we compute the (estimated) smallest $N$ so that the upper bound 
 of $ P(|S_{\bf u, N}-S^{\bf u}(f)| \geq t) $ of Corollary \ref{cor:concentr_pire} achieves 5\% (i.e., the sum of the right-hand sides of \eqref{eq:concentr_pire1} and \eqref{eq:concentr_pire2} is less than $0.05$). The constant $V$ is estimated empirically. The results of these computations are displayed in Figure \ref{fig:concentr}. The set ${\bf u}$ is $\{1 \}$ or $\{ 2 \}$.

\begin{figure}
\centering
\includegraphics[scale=.5]{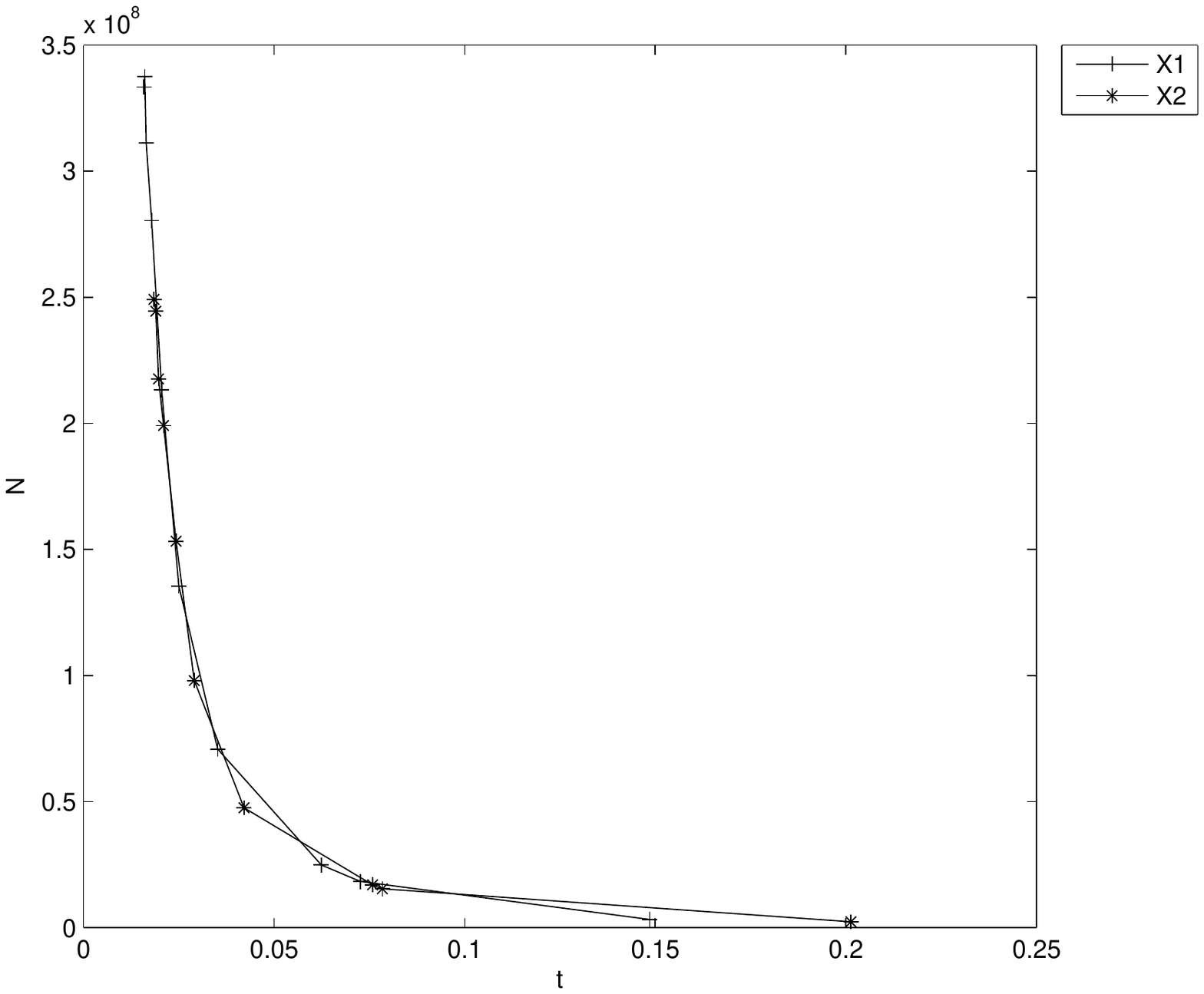}
\caption{(Estimated) smallest $N$ to have $ P(|S_{\bf u, N}-S^{\bf u}(f)| \geq t)\leq 5\% $ for ${\bf u}=\{1\}$ or $\{2\}$. }
\label{fig:concentr}
\end{figure}

These results clearly show that the confidence intervals produced by the use of the concentration inequality
 on $S^{\bf u}(f)$ require a large sample size. As a consequence, its use is only possible when many evaluations of the output function $f$ are available.

\subsubsection{Mass-spring model}
In this section, we consider the displacement $x(t)$ of a mass connected to a spring for $t \in [0;40]$. This displacement is given by the following second-order differential equation
\[ m x''(t) + c x'(t) + k x(t) = 0, \]
together with initial conditions $x(0)=l$, $x'(0)=0$. We use the readily-available analytical closed-form expression of this initial-value problem for $x(t)$.

The input parameters are $X=(m,c,k,l)$ (so that $p=4$) whose interpretations and distributions are given in Table \ref{t:2}.

The output vector is defined by
\[ Y = f(X) = \left( x(t_1), x(t_2), \ldots, x(t_{800}) \right),\quad \text{ for }\quad t_i=0.05 i \quad \text{ and }\quad k=800.\]

\begin{table}
\centering
\begin{tabular}{|c|c|l|}
\hline
Variable & Interpretation (SI unit) & \quad Distribution \\ \hline
$m$ & mass ($\text{kg}$) & \text{Unif}([10;12]) \\ 
$c$ & damping constant ($\text{N}\cdot \text{m}^{-1}\cdot \text{s}$) & \text{Unif}([0.4; 0.8]) \\ 
$k$ & spring constant ($\text{N} \cdot \text{m}^{-1}$) & \text{Unif}([70;90]) \\ 
$l$ & initial elongation ($\text{m}$) & \text{Unif}([-1; -0.25]) \\ \hline
\end{tabular}
\caption{Interpretations and distributions of the parameters in the spring-mass model. }
\label{t:2}
\end{table}

\paragraph*{Unidimensional first-order Sobol indices}
By considering each component of $Y$ independently, it is possible to estimate the (unidimensional first-order) Sobol indices of $Y(t_i)$ for $i=1,\ldots,800$ and each input variable. This gives the plot of Figure \ref{fig:ponctSob}. 

\begin{figure}
\centering
\includegraphics[scale=.8]{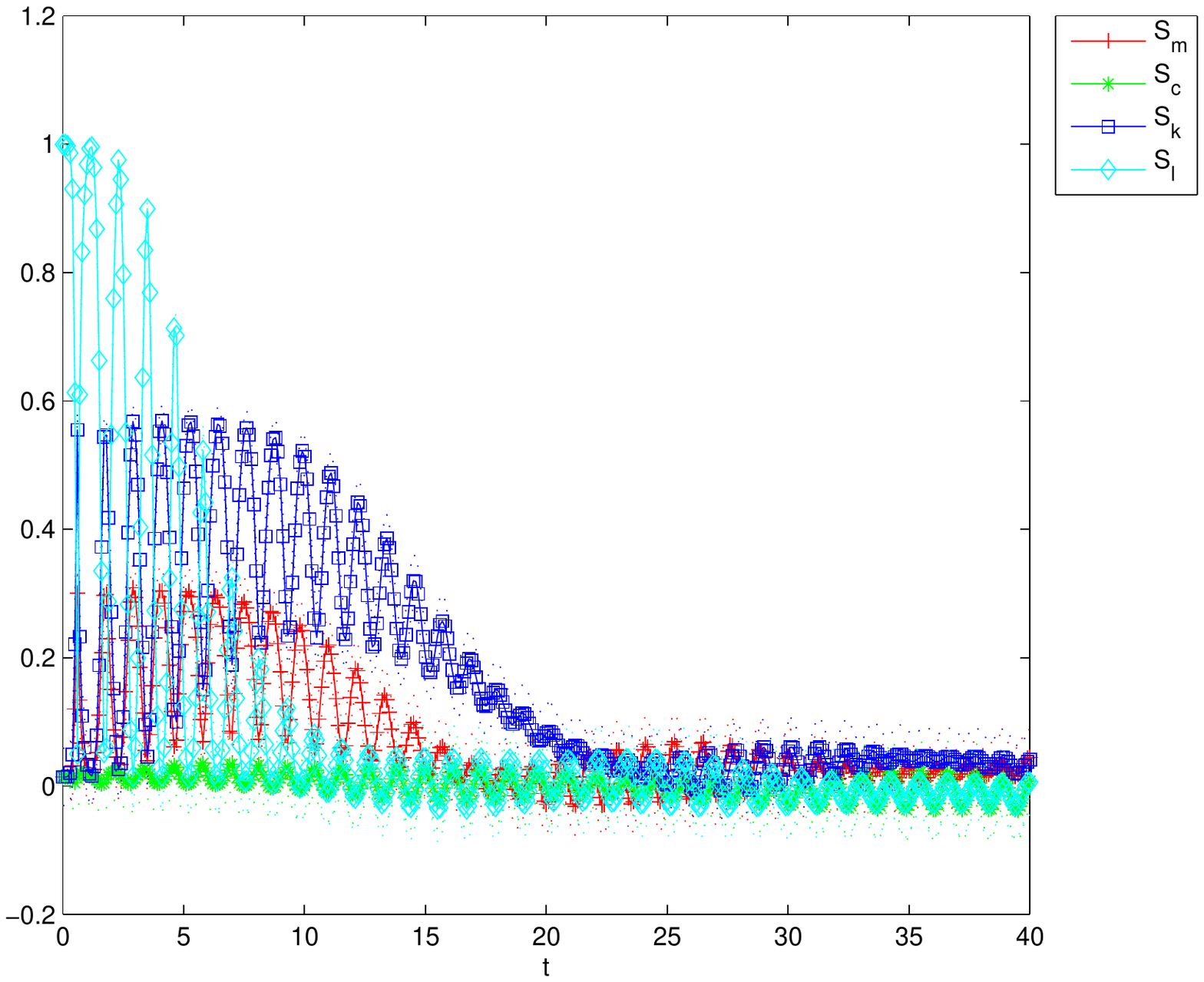}
\caption{Plots of the estimated unidimensional first-order Sobol indices $S^{\bf u}(Y(t))$ as functions of $t$ for ${\bf u}=\{m\}$, $\{c\}$, $\{k\}$ or $\{l\}$. The dots around each curve form the hull of the 95\% confidence intervals for these indices with $N=2000$.}
\label{fig:ponctSob}
\end{figure}

This plot seems difficult to interpret since we can see that the indices for $l$, $k$ and $m$ oscillate rapidly, leading to a frequent change of their respective rankings as time evolves. This is an additional motivation for using the generalized Sobol indices considered in this paper, easier to interpret. Note that, for large values of $t$, the first-order indices do not sum up to 1 meaning that interactions between the variables have a large influence for such $t$'s.

\paragraph*{Generalized Sobol indices}
We have computed the generalized Sobol indices for the output vector $Y$, for ${\bf u}=\{m\}$, $\{c\}$, $\{k\}$ or $\{l\}$ as well as their 95\% confidence intervals for $N=2000$. The numerical results are gathered in Table \ref{t:gensob}.

\begin{table}
\centering
\begin{tabular}{|c|c|c|}
\hline
Variable $\bf u$ & Punctual estimate for $S^{\bf u}(Y)$ & 95\%  confidence interval for $S^{\bf u}(Y)$ \\ \hline
$m$ & 0.0826 & [0.0600 ; 0.1052] \\ 
$c$ & 0.0020 & [-0.0181; 0.0222] \\ 
$k$ & 0.2068 & [0.1835 ; 0.2301] \\ 
$l$ & 0.0561 & [0.0328 ; 0.0794] \\ \hline
\end{tabular}
\caption{Results of the estimation of the first-order generalized Sobol indices in the spring-mass model. }
\label{t:gensob}
\end{table}

This computation makes clear that the ranking of the first-order influence indices of each input parameter is $S^{\{k\}}(Y) > S^{\{m\}}(Y) > S^{\{l\}}(Y) > S^{\{c\}}(Y)$.

\section{Case of functional outputs}\label{sec:functional}
In many practical situations the output $Y$ is  functional. It is then useful to extend the vectorial indices to functional outputs. This is the aim of the following section.

\subsection{Definition}

Let $\mathbb H$ be a separable Hilbert space endowed with the scalar product $\langle \cdot,\cdot \rangle$ and the norm $||\cdot||$. Let $f$ be a $\mathbb H$-valued function, i.e. $Y$ and $Y^{\bf u}$ are $\mathbb H$-valued  random variable. 
We assume that $\mathbb E \left(\|Y\|^2\right)<\infty$. Recall that $\mathbb E \left(Y\right)$ is defined by duality as the unique member of $\mathbb{H}$ satisfying 
$$\mathbb E \left(\langle h,Y\rangle\right)=\langle h,\E(Y)\rangle \quad \trm{for all} \quad h\in \mathbb{H}.$$
Recall that the covariance operator associated with $Y$ is the endomorphism $\Gamma$ on $\mathbb{H}$ defined, for  $h\in \mathbb{H}$ by $\Gamma(h)=\mathbb E \cro{ \langle Y,h\rangle Y}$ . We also recall that it is a well known fact that  $\mathbb E \left(\|Y\|^2\right)<\infty$ implies that $\Gamma$ is then a Trace class operator and its trace is then well defined.
We generalize the definition  of $S^{\bf u}(f)$ introduced in Section \ref{sec:definition} for functional outputs:

\begin{defi}
$S^{\bf u,\infty}(f)= \frac{\Tr(\Gamma_{\bf u})}{\Tr(\Gamma)},$ where  $\Gamma_{\bf u}$ is the endomorphism on $\mathbb{H}$ defined by $\Gamma_{\bf u}(h)=\mathbb E \cro{ \langle Y^{\bf u},h\rangle Y}$ for any $h\in \mathbb{H}$.
\end{defi}

In the next lemma we give the so-called polar decomposition of the traces of $\Gamma$ and $\Gamma_{\bf u}$.

\begin{lem}
We have
\begin{align*}
\Tr(\Gamma)&=\E\left(\|Y\|^2\right)-\|\E(Y)\|^2\\
 \Tr(\Gamma_{\bf u})&=\frac 14 \cro{\E\left(\|Y+Y^{\bf u}\|^2 \right)-\E\left(\|Y-Y^{\bf u}\|^2\right)-4\|\E\left(Y\right)\|^2}.
\end{align*}
\end{lem}
\noi
Let $(\varphi_l)_{1\leq l}$ be an  orthonormal basis of $\mathbb H$. Then
\[
\|Y\|^2=\sum_{i=1}^{\infty} \langle Y,\varphi_i \rangle^2.
\]
\noi
Now, in view of estimation, we truncate the previous sum by setting 
\[
\|Y\|^2_{m}=\sum_{i=1}^{{m}} \langle Y,\varphi_i \rangle^2.
\]
\begin{rmk}
It amounts to truncate the expansion of $Y$ to a certain level ${m}$. Let $Y_{m}$ be the truncated approximation of $Y$:
\[iid
Y_{m}=\sum_{l=1}^{m}  \langle Y,\varphi_i \rangle \varphi_l,
\]
seen as a vector of dimension ${m}$, and results of Section \ref{sec:estimation} can be applied to $Y_m$. Notice that $Y_{m}$ is than the projection of $Y$ onto Span$\left(\varphi_1,\ldots, \varphi_m\right)$.
\end{rmk}

\subsection{Estimation of $S^{\bf u,\infty}(f)$}
As in Section \ref{sec:estimation}, we define the following estimator of  $S^{\bf u,\infty}(f)$:

$$
S_{{\bf u},{m},N} = \frac{\frac1{4N}\sum_{i=1}^{N}\left( \|Y_{i}+Y_{i}^{\bf u}\|_{{m}}^2-\|Y_{i}-Y_{i}^{\bf u}\|_{{m}}^2-\|\overline{Y}+\overline{Y^{\bf u}}\|_{{m}}^2\right)}
{\frac1{N}\sum_{i=1}^{N}\left( \frac{\|Y_{i}\|_{{m}}^2+\|Y_{i}^{\bf u}\|_{{m}}^2}2-\left\|\frac{\overline{Y}+\overline{Y^{\bf u}}}{2}\right\|_{{m}}^2\right)}  .   
$$

Let $T$ be a $\mathbb H$-valued random variable. For any sequence $(T_{i})_{i\in\N^{*}}$ of iid variables distributed as $T$,  we define 
$$D_{N, m}(T)=\frac1{N}\sum_{i=1}^{N}\left( \|T_{i}\|_{{m}}^2-\left\|\overline{T}\right\|_{{m}}^2\right)$$
and
\begin{align*}
e_{j}&=\E\left(\langle T_{i},\varphi_{j}\rangle\right) \\
v_{j}&=\E\left(\langle T_{i},\varphi_{j}\rangle^{2}\right)\\
Z_{i,j}&=\langle T_{i},\varphi_{j}\rangle-e_{j}\\
W_{i,j}&=\langle T_{i},\varphi_{j}\rangle^{2}-v_{j}.
\end{align*}

In the spirit of \cite{fort2012estimation}, we decompose $D_{N, m}(T)$ and give asymptotics for each of the terms of the decomposition.

\begin{proposition}\label{prop:DN}
$\;$
\begin{enumerate}
\item $D_{N, m}(T)$ can be rewritten as the sum of a totally degenerated U-statistic of order 2, a centered linear
 term and a deterministic term in the following way
\begin{equation}\label{decomp_est}
D_{N, m}(T)-\E\left(\| T\|^{2}\right)+\| \E\left(T\right)\|^{2}=-U_NK(T)+P_NL(T)-B_{{m}}(T)
\end{equation}
where
{\setlength\arraycolsep{2pt}
\begin{eqnarray*}
U_NK(T)&:=&\sum_{l=1}^{{m}}\frac{1}{N^{2}}\sum_{1\leq i\neq j\leq N}Z_{i,l}Z_{j,l}\\
P_NL(T)&:=&\frac{1}{N}\left(1-\frac{1}{N}\right)\sum_{l=1}^{{m}}\sum_{i=1}^{ N}\left(W_{i,l}-2e_{l}Z_{i,l}\right)\\
B_{m}(T)&:=&\sum_{l>{m}}\left(v_{l}-e_l^{2}\right)+\frac 1N \sum_{l=1}^{{m}} \left(v_{l}-e_l^{2}\right).\\
\end{eqnarray*}}
\item Assume that there exists $\delta >1$ so that
\begin{equation}\label{eq:hyp_T}
v_l= \E\left(\langle T,\varphi_{l}\rangle^{2}\right)=O(l^{-(\delta+1)}) 
\end{equation}iid
and $\delta'>1$ so that
\begin{equation}\label{eq:hyp_T4}
 \E\left(\pscal{T,\phi_l}^4\right)=O(l^{-\delta'}).
\end{equation}
Then for any $m=m(N)$ so that:
\begin{equation}\label{eq:choixm}
\frac{m(N)}{N^{\frac{1}{2\delta}}} \rightarrow +\infty , \;\;  \frac{m(N)}{\sqrt N} \rightarrow 0,
\end{equation}
we have
\begin{enumerate}
\item $ B_{m}^2(T)=o\left(1/N\right)$
\item $\E\left((U_NK(T))^2\right)=o\left(1/N\right)$
\item $ P_NL(T)-P_NL'(T)=o_{\P}\left(\frac{1}{\sqrt{N}}\right)$
\end{enumerate}
where $P_NL'(T):=\frac{1}{N}\left(1-\frac 1N\right)\sum_{l=1}^{\infty}\sum_{i=1}^{N}\Big[W_{i,l}-2e_lZ_{i,l}\Big]$.
\end{enumerate}
\end{proposition}

\begin{proof}
In order to simplify the notation, we set $U_NK:=U_NK(T)$,\\ $P_NL:=P_NL(T),\  B_{m}:=B_{m}(T)$ and $P_NL':=P_NL'(T)$.\\

\noi
\tbf{a) Term $B_m$}\\\\
Since $\sum_{l=1}^{\infty} v_l< + \infty$, $\left(v_l\right)_l$ is bounded, let $K$ be so that $\max_{ l\geq 1} v_l\leq K$.  We have, for sufficiently large $m(N)$ (hence sufficiently large $N$),
\beq
B_{m} &=& \sum_{l>{m(N)}}(v_l-e_l^2) +\frac 1N \sum_{l=1}^{m(N)}(v_l-e_l^2) \\
&\leq& \sum_{l>{m(N)}}v_l +\frac 1N \sum_{l=1}^{m(N)}v_l\\
&\leq& \sum_{l>{m(N)}} v_l + \frac {m(N) K}{N} \\
&\leq& C \sum_{l>{m(N)}} l^{-(\delta+1)}+\frac {m(N)K}N \text{ for a constant C>0}.  
\eeq

Hence,
\begin{eqnarray*}
N B_m^2 &\leq& 2 N C \left( \sum_{l>m(N)} l^{-(\delta+1)} \right)^2 + \frac 2 N m(N)^2 K^2 \\
&\leq& 2 N C \left( \int_{m(N)}^{+\infty} x^{-(\delta+1)} \,\textrm{d} x \right)^2 + \frac 2 N m(N)^2 K^2 \\
&\leq& 2 N C \delta^{-2} m(N)^{-2 \delta} + \frac 2 N m(N)^2 K^2
\end{eqnarray*}
and both terms go to zero when $N \rightarrow +\infty$ by  \eqref{eq:choixm}.
Hence \[ B_{m}^2=o\left(\frac{1}{N}\right). \]

\noi
\tbf{b) Term $U_NK$}\\\\
One has $\mathbb E ((U_NK)^2)=  E_1+ E_2+ E_3 $ where
\begin{eqnarray*}
 E_1&= &\frac{2}{N^4}\sum_{l,k=1}^{m}
 \sum_{1\le i_1\ne j_1\le N} \mathbb{E}  \Big[Z_{i_1,l}Z_{i_1,k}Z_{j_1,l}Z_{j_1,k}\Big],\\
E_2 &=& \frac{4}{N^4} \sum_{l,k=1}^{m}
 \sum_{ i_1, j_1,j_2 all \ne} \mathbb E \Big[ Z_{i_1,l}Z_{i_1,k}Z_{j_1,l}Z_{j_2,k}\Big],\\
 E_3&=& \frac{1}{N^4} \sum_{l,k=1}^{m}
 \sum_{ i_1, j_1,i_2,j_2 all \ne} \mathbb E\Big[Z_{i_1,l}Z_{i_2,k}Z_{j_1,l}Z_{j_2,k} \Big]. \\
\end{eqnarray*}

\noi
One easily see that 
$  E_2=E_3=0 $, since, for all $l$,  the variables $(Z_{i,l})_{1\leq i\leq N} $ are centered and independent.\\\\
Let us now compute and bound $ E_1$.
\begin{eqnarray*}
 E_1&= &\frac{2}{N^4}\sum_{l,k=1}^{m} 
\sum_{1\le i_1\ne j_1\le N} \mathbb{E}  \Big[Z_{i_1,l}Z_{i_1,k}Z_{j_1,l}Z_{j_1,k}\Big]\\
&=& \frac{2}{N^2}\left(1-\frac 1N\right)\sum_{l,k=1}^{m} \mathbb{E}  \Big[Z_{1,l}Z_{1,k}Z_{2,l}Z_{2,k}\Big]\\
&=& \frac{2}{N^2}\left(1-\frac 1N\right)\sum_{l,k=1}^{m} \Big[ \mathbb{E}  \left(Z_{1,l}Z_{1,k}\right)\Big]^2\\
&\leq& \frac{2}{N^2}\left(1-\frac 1N\right)\sum_{l,k=1}^{m}  \mathbb{E}  \left(Z_{1,l}^2\right)\mathbb{E}\left(Z_{1,k}^{{2}}\right)\\ 
&\leq& \frac{2}{N^2}\left(1-\frac 1N\right)   \left(\sum_{l=1}^{m}   \mathbb{E} \left( Z_{1,l}^2 \right)\right)^2\\
&\leq& \frac{2}{N^2}\left(1-\frac 1N\right)\left(\sum_{l=1}^{m} v_{l}\right)^2,
\end{eqnarray*}
as
\begin{equation}\label{e:trucutile} 0 \leq \E(Z_{1,l}^2) = \E((\left\langle T_1,\phi_l\right\rangle -e_l)^2) = \E(\left\langle T_1, \phi_l\right\rangle^2) - e_l^2 \leq v_l. \end{equation}
 By assumption \eqref{eq:hyp_T}, the series $\left(\sum_l v_l\right)$ is convergent. Thus we proceed as for $B_m$ to get, for sufficiently large $m(N)$ (hence sufficiently large $N$), 
 \begin{eqnarray*}
 E_1
&\leq& \frac{2}{N^2}\left(1-\frac 1N\right) \left(m(N)K\right)^2,
\end{eqnarray*}
 
As a consequence, by \eqref{eq:choixm}, $ E_1=o\left(\frac{1}{N}\right)$
 and we obtain that $\E\left((U_NK)^2\right)=o\left(\frac{1}{N}\right)$.\\\\

\noi\\
\tbf{c) Term $P_NL$}\\\\
By Markov inequality we have
\begin{equation}\label{eq:nov3}
\P\left( \sqrt N \left|P_NL'-P_NL\right|>\epsilon\right)\leq \frac{ N}{ \epsilon^{2}}\E\left( \left|P_NL'-P_NL\right|^{2}\right).
\end{equation}
Hence, it is sufficient to prove that  $N\E\left( \left|P_NL'-P_NL\right|^{2}\right) \rightarrow 0$ when $N\rightarrow+\infty$. But
 \begin{eqnarray*}
  \left(P_NL'-P_NL\right)^{2}&=&  \left(\frac{1}{N}\left(1-\frac 1N\right)\sum_{i=1}^{N} \sum_{l>{m}} \Big[W_{i,l}-2e_lZ_{i,l}\Big]\right)^{2}\\
 &\leq&\left( \frac{1}{N}\sum_{i=1}^{N} A_{m,i} \right)^2, 
 \end{eqnarray*}
where $A_{m,i}:=\sum_{l>{m}} ( W_{i,l}-2e_lZ_{i,l})$.  Then
 \begin{eqnarray}
N\E\left(  \left(P_NL'-P_NL\right)^{2}\right) 
&\leq &  \frac{1}{N}\Var \left( \sum_{i=1}^{N} A_{m,i} \right)
=\Var \left(  A_{m,1} \right). \label{eqnov2}
 \end{eqnarray}
The last inequalities come from the fact that $(A_{m,i})_{i=1,\ldots,N}$ are centered and i.i.d r.v.s. Indeed, since by assumption 
\eqref{eq:hyp_T},
we can apply Tonelli's theorem to show that
$ \sum_{l>m} \E\left( \abs{W_{i,l}} \right)$ and $
   \sum_{l>m} \E\left( \abs{Z_{i,l}} \right)$
are finite. Hence, by Fubini's theorem and the fact that each variable $W_{i,l}-2e_lZ_{i,l}$ is centered, we get
$$\E(A_{m,i})=\E\left(\sum_{l>{m}} ( W_{i,l}-2e_lZ_{i,l})\right)
=\sum_{l>{m}} \E\left( W_{i,l}-2e_lZ_{i,l} \right)=0,$$
which proves that $A_{m,i}$ is centered. \\

It remains now to upper-bound $\Var(A_{m,1})$. 
\begin{eqnarray*}
 \Var(A_{m,1})&=&\E\left( \left(\sum_{l> {m}} ( W_{1,l}-2e_lZ_{1,l}) \right)^2 \right)\\
 &\leq&2\E\left(\left(\sum_{l> {m}}  W_{1,l} \right)^2 \right)+2\E\left(\left(\sum_{l> {m}}2e_lZ_{1,l}\right)^2\right). \\ 
  \end{eqnarray*}
On one hand, for all sufficiently large $m$,
\begin{eqnarray*}
\E\left(\left(\sum_{l> {m}}  W_{1,l} \right)^2 \right)&=&\E\left(\left(\sum_{l> {m}}  \pscal{T,\phi_l}^2-\E(\pscal{T,\phi_l}^2) \right)^2 \right)\\
 &=& \E\left( \left( \sum_{l>m} \pscal{T,\phi_l}^2\right)^2\right) - \left( \sum_{l>m} \E(\pscal{T,\phi_l}^2)\right)^2  \\
&\leq& \E\left( \left(\sum_{l>m}\pscal{T,\phi_l}^2\right)^2\right) \\
&=& \sum_{l>m, l'>m} \E\left( \pscal{T,\phi_l}^2 \pscal{T,\phi_{l'}}^2\right) \\
&\leq& \sum_{l>m, l'>m} \sqrt{ \E\left(\pscal{T,\phi_l}^4\right) \E\left(\pscal{T,\phi_{l'}}^4\right) } \text{ (Cauchy-Schwarz)}\\
&=& \left( \sum_{l>m} \sqrt{\E\left(\pscal{T,\phi_l}^4\right)} \right)^2 \\
&\leq& \sum_{l>m} \E\left(\pscal{T,\phi_l}^4\right) \text{ (Jensen's inequality)} \\
&\rightarrow& 0 \text{ when } N \rightarrow +\infty, \text{ by } \eqref{eq:hyp_T4} \text{ and } m(N)\rightarrow+\infty
\end{eqnarray*}
and on the other hand,
\begin{eqnarray*}
\E\left(\left(\sum_{l>m} e_l Z_{1,l}\right)^2\right) &\leq& \E\left(\sum_{l>m} e_l^2 \sum_{l>m} Z_{1,l}^2\right) \text{ (Cauchy-Schwarz)}\\
&\leq& \left( \sum_{l>m} e_l^2 \right) \left( \sum_{l>m} v_l \right) \text{ by }\eqref{e:trucutile} \\
&\leq& \left( \sum_{l>m} v_l \right)^2 \text{ (because } e_l^2 \leq v_l \text{ by Jensen's inequality)} \\
&\rightarrow& 0 \text{ when } N \rightarrow +\infty \text{ by } \eqref{eq:hyp_T}.
\end{eqnarray*}
So  $\Var(A_{m,1})=o(1)$, hence (by \eqref{eqnov2} and \eqref{eq:nov3}) we have $P_NL-P_NL'=o_{\P}\left(\frac{1}{\sqrt{N}}\right)$.

\end{proof}

\begin{thm}
Suppose that conditions \eqref{eq:hyp_T}, \eqref{eq:hyp_T4} and \eqref{eq:choixm} of Proposition \ref{prop:DN} are fulfilled. Then we have:
\begin{equation} \label{TLC}
\sqrt{N}(S_{{\bf u},{m},N}-S^{\bf u}(f))\overset{\mathcal{L}}{\underset{N\to\infty}{\rightarrow}} \mathcal{N}(0,\sigma^{2})
\end{equation}
with  $\sigma^{2}$ depending on the moments of $\left(D_{N, m}(Y), D_{N, m}(Y+Y^{\bf u}),D_{N, m}(Y-Y^{\bf u})\right).$ 
\end{thm}

\noi 
Before starting the proof of the theorem we state an auxiliary lemma:

\begin{lem}\label{lem:control}
$
\sqrt N \left\|\overline{Y}-\overline{Y^{\bf u}}\right\|_{{m}}^2\
$
converges to $0$ in probability.
\end{lem}
\begin{proof}
Since
$$
\sqrt N \left\|\overline{Y}-\overline{Y^{\bf u}}\right\|_{{m}}^2\leq\sqrt N \left\|\overline{Y}-\overline{Y^{\bf u}}\right\|^2=\frac{1}{\sqrt{N}}\left\|\sqrt N \left[\frac1N\sum_{i=1}^{N}\left(Y_{i}-Y_{i}^{\bf u}\right)\right]\right\|^{2}.
$$
We conclude using the central limit theorem for random variables valued in an Hilbert space (see e.g. \cite{TL91}).
\end{proof}

\begin{proof}
\noi First we note that
$$
S_{{\bf u},{m},N}=\frac12\frac{D_{N, m}(Y+Y^{\bf u})-D_{N, m}(Y-Y^{\bf u})-\left\|\overline{Y}-\overline{Y^{\bf u}}\right\|_{{m}}^2}{D_{N, m}(Y)+D_{N, m}(Y^{\bf u})+\left\|\overline{Y}-\overline{Y^{\bf u}}\right\|_{{m}}^2}.
$$
The proof will be decomposed into 3 steps.\\

$\bullet$ \begin{bf} Step 1 \end{bf} We  prove a vectorial central limit theorem (CLT) for 
\[ \left(D_{N, m}(Y),D_{N, m}(Y^{\bf u}), D_{N, m}(Y+Y^{\bf u}),D_{N, m}(Y-Y^{\bf u})\right). \]
 The vector
$$\mathbb V_{N}:=\left(D_{N, m}(Y),D_{N, m}(Y^{\bf u}), D_{N, m}(Y+Y^{\bf u}),D_{N, m}(Y-Y^{\bf u})\right)$$ can be decomposed as in Proposition  \ref{prop:DN} in the following way

\begin{equation}
\mathbb V_{N}-\E\left(\mathbb V_{N}\right)=-U_N\mathbb K+P_N\mathbb L'+\left(P_N\mathbb L-P_N\mathbb L'\right)-\mathbb B_{{m}}
\end{equation}

where
{\setlength\arraycolsep{2pt}
\begin{eqnarray*}
U_N\mathbb K&:=&\left(U_{N}K(Y),U_{N}K(Y^{\bf u}),U_{N}K(Y+Y^{\bf u}),U_{N}K(Y-Y^{\bf u})\right)\\
P_N\mathbb L&:=&\left(P_{N}L(Y),P_{N}L(Y^{\bf u}),P_{N}L(Y+Y^{\bf u}),P_{N}L(Y-Y^{\bf u})\right)\\
P_N\mathbb L'&:=&\left(P_{N}L'(Y),P_{N}L'(Y^{\bf u}),P_{N}L'(Y+Y^{\bf u}),P_{N}L'(Y-Y^{\bf u})\right)\\
\mathbb B_{m}&:=&\left( B_{m}(Y) B_{m}(Y^{\bf u}), B_{m}(Y+Y^{\bf u}), B_{m}(Y-Y^{\bf u})\right).\\
\end{eqnarray*}}
By Proposition  \ref{prop:DN} 2., it is enough to prove a CLT for $P_N\mathbb L'$. This is the case  since it is an empirical sum of i.i.d. centered random vectors. \\

$\bullet$ \begin{bf} Step 2 \end{bf} Using Lemma \ref{lem:control}  and the Delta method, we derive  a  CLT for 
$$\left(D_{N, m}(Y+Y^{\bf u})-D_{N, m}(Y-Y^{\bf u})-\left\|\overline{Y}-\overline{Y^{\bf u}}\right\|_{{m}}^2, D_{N, m}(Y)+D_{N, m}(Y^{\bf u})+\left\|\overline{Y}-\overline{Y^{\bf u}}\right\|_{{m}}^2\right).$$

$\bullet$ \begin{bf} Step 3 \end{bf} We conclude using the Delta method. 
\end{proof}

\paragraph{Acknowledgements}
This work has been partially supported by the French National
Research Agency (ANR) through COSINUS program (project COSTA-BRAVA
no ANR-09-COSI-015). The authors are grateful to Herv\'e Monod and Cl\'ementine Prieur for fruitful discussions.

\bibliographystyle{plain}

\begin{thebibliography}{10}

\bibitem{bobor2007}
E~Borgonovo.
\newblock A new uncertainty importance measure.
\newblock {\em Reliability Engineering \& System Safety}, 92(6):771--784, 2007.

\bibitem{campbell2006sensitivity}
Katherine Campbell, Michael~D McKay, and Brian~J Williams.
\newblock Sensitivity analysis when model outputs are functions.
\newblock {\em Reliability Engineering \& System Safety}, 91(10):1468--1472,
  2006.

\bibitem{Fort2013}
J.-C. {Fort}, T.~{Klein}, and N.~{Rachdi}.
\newblock {New sensitivity analysis subordinated to a contrast}.
\newblock {\em ArXiv e-prints}, May 2013.

\bibitem{fort2012estimation}
J.C. Fort, T.~Klein, A.~Lagnoux, and B.~Laurent.
\newblock Estimation of the sobol indices in a linear functional
  multidimensional model.
\newblock 2012.

\bibitem{pickfreeze}
Fabrice Gamboa, Alexandre Janon, Thierry Klein, Agn\`es Lagnoux-Renaudie, and
  Cl{\'e}mentine Prieur.
\newblock Statistical inference for sobol pick freeze monte carlo method.
\newblock Preprint available at \texttt{http://hal.inria.fr/hal-00804668/en},
  2013.

\bibitem{janon2012asymptotic}
Alexandre Janon, Thierry Klein, Agn\`es Lagnoux-Renaudie, Ma{\"e}lle Nodet, and
  Cl{\'e}mentine Prieur.
\newblock Asymptotic normality and efficiency of two sobol index estimators.
\newblock {\em To appear in ESAIM P\& S}, 2013.

\bibitem{lamboni2011multivariate}
Matieyendou Lamboni, Herv{\'e} Monod, and David Makowski.
\newblock Multivariate sensitivity analysis to measure global contribution of
  input factors in dynamic models.
\newblock {\em Reliability Engineering \& System Safety}, 96(4):450--459, 2011.

\bibitem{Ledoux01}
Michel Ledoux.
\newblock {\em The concentration of measure phenomenon}, volume~89 of {\em
  Mathematical Surveys and Monographs}.
\newblock American Mathematical Society, Providence, RI, 2001.

\bibitem{TL91}
Michel Ledoux and Michel Talagrand.
\newblock {\em Probability in Banach Spaces. Isoperimetry and Processes}.
\newblock Springer, Berlin, A991.

\bibitem{Owen2013}
A.~{Owen}, J.~{Dick}, and S.~{Chen}.
\newblock {Higher order Sobol' indices}.
\newblock {\em ArXiv e-prints}, June 2013.

\bibitem{Owen2012}
A.~B. {Owen}.
\newblock {Variance components and generalized Sobol' indices}.
\newblock {\em ArXiv e-prints}, May 2012.

\bibitem{saltelli2008global}
A.~Saltelli, M.~Ratto, T.~Andres, F.~Campolongo, J.~Cariboni, D.~Gatelli,
  M.~Saisana, and S.~Tarantola.
\newblock {\em Global sensitivity analysis: the primer}.
\newblock Wiley Online Library, 2008.

\bibitem{sobol1993}
I.~M. Sobol.
\newblock Sensitivity estimates for nonlinear mathematical models.
\newblock {\em Math. Modeling Comput. Experiment}, 1(4):407--414 (1995), 1993.

\bibitem{van2000asymptotic}
A.~W. van~der Vaart.
\newblock {\em Asymptotic statistics}, volume~3 of {\em Cambridge Series in
  Statistical and Probabilistic Mathematics}.
\newblock Cambridge University Press, Cambridge, 1998.

\bibitem{xu2011}
Chonggang Xu and George~Zdzislaw Gertner.
\newblock Reliability of global sensitivity indices.
\newblock {\em Journal of Statistical Computation and Simulation},
  81(12):1939--1969, 2011.

\end{thebibliography}

\def\cprime{$'$}

\end{document}